\documentclass[11pt]{article}
\usepackage{amsmath,amsthm,amssymb,epsfig,sectsty}
\usepackage[hmargin=1in,vmargin=1in]{geometry}

\long\def\beginpgfgraphicnamed#1#2\endpgfgraphicnamed{\includegraphics{#1}}
\usepackage{graphicx}

\theoremstyle{plain}
\newtheorem{theorem}{Theorem}[section]
\newtheorem{lemma}[theorem]{Lemma}

\theoremstyle{definition}

\newtheorem{example}[theorem]{Example}

\newcommand{\qedsymb}{\hfill{\rule{2mm}{2mm}}}
\renewenvironment{proof}{\begin{trivlist} \item[\hspace{\labelsep}{\bf \noindent Proof.\/}] }{\qedsymb\end{trivlist}}%
\newcommand{\bs}[1]{\boldsymbol{#1}}

\newcommand{\expar}[2]{\mathrm{E}_{ #1 } [ #2 ]}
\newcommand{\bbR}{\mathbb{R}}
\newcommand{\Z}{{\mathbb{Z}}}
\newcommand{\eps}{\varepsilon}

\newcommand{\R}{{\cal{R}}}

\newcommand{\C}{{\cal{C}}}

\newcommand{\polylog}{\mathrm{polylog}}
\newcommand{\opt}{\mbox{\textsc{opt}}}

\sectionfont{\large} \subsectionfont{\normalsize}%
\numberwithin{equation}{section}%
\newcommand{\Xomit}[1]{ }

\begin{document}

\title{Improved Approximation Guarantees for Weighted Matching \\
in the Semi-Streaming Model}
\author{%
Leah Epstein\thanks{Department of Mathematics, University of Haifa, 31905 Haifa, Israel.
Email: {\tt lea@math.haifa.ac.il}.}%
\and Asaf Levin\thanks{Chaya fellow. Faculty of Industrial Engineering and Management, The Technion,
32000 Haifa, Israel. Email: {\tt levinas@ie.technion.ac.il}.}%
\and Juli\'{a}n Mestre\thanks{ Max-Planck-Institut f\"ur Informatik, 66123 Saarbr\"ucken, Germany.
Email: {\tt jmestre@mpi-inf.mpg.de}. Research supported by an Alexander von Humboldt Fellowship.}%
\and Danny Segev\thanks{Operations Research Center, Massachusetts Institute of Technology, Cambridge 02139, MA, USA. Email: {\tt segevd@mit.edu}.}}
\date{}
\maketitle

\begin{abstract}
  We study the maximum weight matching problem in the semi-streaming
  model, and improve on the currently best one-pass algorithm due to
  Zelke (Proc.\ STACS~'08, pages 669--680) by devising a
  deterministic approach whose performance guarantee is $4.91 +
  \eps$. In addition, we study {\em preemptive} online algorithms, a
  sub-class of one-pass algorithms where we are only allowed to
  maintain a feasible matching in memory at any point in time. All
  known results prior to Zelke's belong to this sub-class. We provide
  a lower bound of $4.967$ on the competitive ratio of any such
  deterministic algorithm, and hence show that future improvements
  will have to store in memory a set of edges which is not necessarily
  a feasible matching.
\end{abstract}

\section{Introduction} \label{sec:intro}

The computational task of detecting maximum weight matchings is one of
the most fundamental problems in discrete optimization, attracting
plenty of attention from the operations research, computer science,
and mathematics communities. (For a wealth of references on matching
problems see~\cite{Schrijver03}.) In such settings, we are given an
undirected graph $G = (V,E)$ whose edges are associated with
non-negative weights specified by $w : E \to \bbR_+$. A set of edges
$M \subseteq E$ is a {\em matching} if no two of the edges share a
common vertex, that is, the degree of any vertex in $(V,M)$ is at most
$1$. The weight $w(M)$ of a matching $M$ is defined as the combined
weight of its edges, i.e., $\sum_{e \in M} w(e)$. The objective is to
compute a matching of maximum weight.  We study this problem in two
related computational models: the {\it semi-streaming} model and the
{\it preemptive online} model.

\paragraph{The semi-streaming model.} Even though these settings
appear to be rather simple as first glance, it is worth noting that
matching problems have an abundance of flavors, usually depending on
how the input is specified. In this paper, we investigate
weighted matchings in the {\em semi-streaming} model, was first
suggested by Muthukrishnan~\cite{Muthukrishnan05}. Specifically, a
{\em graph stream} is a sequence $e_{i_1}, e_{i_2}, \ldots$ of
distinct edges, where $e_{i_1}, e_{i_2}, \ldots$ is an arbitrary
permutation of $E$. When an algorithm is processing the stream, edges
are revealed sequentially, one at a time. Letting $n = |V|$ and $m =
|E|$, efficiency in this model is measured by the space $S(n,m)$ a
graph algorithm uses, the time $T(n,m)$ it requires to process each
edge, and the number of passes $P(n,m)$ it makes over the input
stream. The main restriction is that the space $S(n,m)$ is limited to
$O( n \cdot \polylog(n) )$ bits of memory.  We refer the reader to a
number of recent papers~\cite{Muthukrishnan05, FeigenbaumKMSZ05,
  FeigenbaumKMSZ08, ElkinZ06, McGregor05} and to the references
therein for a detailed literature review.

\paragraph{Online graph problems.} Unlike the semi-streaming model, in
online problems the size of the underlying graph is not known in
advance. The online matching problem has previously been modeled as
follows. Edges are presented one by one to the algorithm, along with
their weight. Once an edge is presented, we must make an irrevocable
decision, whether to accept it or not. An edge may be accepted only if
its addition to the set of previously accepted edges forms a feasible
matching. In other words, an algorithm must keep a matching at all
times, and its final output consists of all edges which were ever
accepted. In this model, it is easy to verify that the competitive
ratio of any (deterministic or randomized) algorithm exceeds any
function of the number of vertices, meaning that no competitive
algorithm exists. However, if all weights are equal, a greedy approach
which accepts an edge whenever possible, has a competitive ratio of
$2$, which is best possible for deterministic algorithms \cite{KVV90}.

Similarly to other online settings (such as call control problems
\cite{GG+}), a preemptive model can be defined, allowing us to remove
a previously accepted edge from the current matching at any point in
time; this event is called {\em preemption}. Nevertheless, an edge
which was either rejected or preempted cannot be inserted to the
matching later on. We point out that other types of online matching
problems were studied as well
\cite{KVV90,KalP93,KMV94,BansalBGN07}.

\paragraph{Comparison between the models.} Both semi-streaming
algorithms and online algorithms perform a single pass over the
input. However, unlike semi-streaming algorithms, online algorithms
are allowed to concurrently utilize memory for two different
purposes. The first purpose is obviously to maintain the current
solution, which must always be a feasible matching, implying that the
memory size of this nature is bounded by the maximal size of a
matching. The second purpose is to keep track of arbitrary information
regarding the past, without any concrete bound on the size of memory
used. Therefore, in theory, online algorithms are allowed to use much
larger memory than is allowed in the semi-streaming model. Moreover, although
this possibility is rarely used, online algorithms may perform
exponential time computations whenever a new piece of input is
revealed. On the other hand, a semi-streaming algorithm may re-insert an edge the
current solution, even if it has been temporarily removed, as long as
this edge was kept in memory.  This extra power is not allowed for
online (preemptive) algorithms, making them inferior in this sense in
comparison to their semi-streaming counterparts.

\paragraph{Previous work.} Feigenbaum et al.~\cite{FeigenbaumKMSZ05}
were the first to study matching problems under similar
assumptions. Their main results in this context were a
semi-streaming algorithm that computes a $(3/2 -
\eps)$-approximation in $O( \log(1 / \eps) / \eps )$ passes for
maximum cardinality matching in bipartite graphs, as well as a
one-pass $6$-approximation for maximum weighted matching in
arbitrary graphs. Later on, McGregor~\cite{McGregor05} improved
on these findings, to obtain performance guarantees of $(1 +
\eps)$ and $(2 + \eps)$ for the maximum cardinality and maximum
weight versions, respectively, being able to handle arbitrary
graphs with only a constant number of passes (depending on $1 /
\eps$). In addition, McGregor~\cite{McGregor05} tweaked the
one-pass algorithm of Feigenbaum et al.\ into achieving a ratio
of $5.828$. Finally, Zelke~\cite{Zelke08} has recently attained
an improved approximation factor of $5.585$, which stands as the
currently best one-pass algorithm. Note that the $6$-approximation
algorithm in~\cite{FeigenbaumKMSZ05} and the
$5.828$-approximation algorithm in~\cite{McGregor05} are
preemptive online algorithms. On the other hand, the algorithm of
Zelke~\cite{Zelke08} uses the notion of shadow-edges, which may be
re-inserted into the matching, and hence it is not an online
algorithm.

\paragraph{Main result I.}
The first contribution of this paper is to improve on the
above-mentioned results, by devising a deterministic one-pass
algorithm in the semi-streaming model, whose performance guarantee
is $4.91 + \eps$. In a nutshell, our approach is based on
partitioning the edge set into $O( \log n )$ weight classes, and
computing a separate maximal matching for each such class in
online fashion, using $O( n \cdot \polylog(n) )$ memory bits
overall. The crux lies in proving that the union of these
matchings contains a single matching whose weight compares
favorably to the optimal one. The specifics of this algorithm are
presented in Section~\ref{sec:app}.

\paragraph{Main result II.} Our second contribution is motivated by
the relation between semi-streaming algorithms and {\em
preemptive} online algorithms, which must maintain a feasible
matching at any point in time. To our knowledge, there are
currently no lower bounds on the competitive ratio that can be
achieved by incorporating preemption. Thus, we also provide a
lower bound of $4.967$ on the performance guarantee of any such
deterministic algorithm. As a result, we show that improved one
pass algorithms for this problem must store more than just a
matching in memory. Further details are provided in
Section~\ref{sec:lower_bound}.

\section{The Semi-Streaming Algorithm}\label{sec:app}

This section is devoted to obtaining main result I, that is, an improved
one-pass algorithm for the weighted matching problem in the
semi-streaming model.  We begin by presenting a simple deterministic algorithm with a performance guarantee of $8$. We then show how to randomize its parameters, still within the semi-streaming framework, and obtain an expected approximation ratio of $4.9108$. Finally,
we de-randomize the algorithm by showing how to emulate the required randomness
using multiple copies (constant number) of the deterministic
algorithm, while paying an additional additive factor of at most $\eps$, for any fixed $\eps > 0$.

\subsection{A simple deterministic approach} \label{det:alg}

\paragraph{Preliminaries.} We maintain the maximum weight of any edge $w_{\max}$ seen so far in
the input stream. Clearly, the maximum weight matching of the edges seen
so far has weight in the interval $[w_{\max},{\frac n2 w_{\max}}]$.
Note that if we disregard all edges with weight at most $2\eps
w_{\max}\over n$, the weight of the maximum weight
matching in the resulting instance decreases by an additive term of
at most $\eps w_{\max} \leq \eps \opt$.

Our algorithm has a parameter $\gamma >1$, and a value $\phi >0$.  We
define weight classes of edges in the following way.  For every $i\in
\Z$, we let the class $W_i$ be the collection of edges whose weight is in the interval $[\phi \gamma^{i},\phi \gamma^{i+1})$. We note that by our initial assumption, the weight of each edge
is in the interval $[{2\eps w_{\max}\over n},w_{\max}]$, and we say
that a weight class $W_i$ is {\it under consideration} if its
weight interval $[\phi \gamma^{i},\phi \gamma^{i+1})$ intersects
$[{2\eps w_{\max}\over n},w_{\max}]$. The number of
classes which are under consideration at any point in time is
$O(\log_{\gamma} (\frac n{\eps}))$.

\paragraph{The algorithm.} Our algorithm simply
maintains the list of classes under consideration and maintains a
maximal (unweighted) matching for each such class. In other words, when
the value of $w_{\max}$ changes, we delete from the memory some
of these matchings, corresponding to the classes which stop being
under consideration.  Note that to maintain a maximal matching in
a given subgraph, we only need to check if the two endpoints of
the new edge are not covered by existing edges of the matching.

To conclude, for every new edge $e\in E$ we proceed as follows.  We
first check if $w(e)$ is greater than the current value of $w_{\max}$.  If so, we update $w_{\max}$ and the list of weight classes under consideration accordingly. Then, we find the weight class of $w(e)$, and try to extend its corresponding matching, i.e., $e$ will be added to
this matching if it remains a matching after doing so.

Note that at each point the content of the memory is the value
$w_{\max}$ and a collection of $O(\log_{\gamma} (\frac n{\eps}))$
matchings, consisting of $ O( n \log_{\gamma} (\frac
n{\eps}))$ edges overall.  Therefore, our algorithm indeed falls in the
semi-streaming model.

At the conclusion of the input sequence, we need to return a single matching rather than a
collection of matchings.  To this end, we could compute a
maximum weighted matching of the edges in the current memory.
However, for the specific purposes of our analysis, we use the following faster
algorithm.  We sort the edges in memory in decreasing order of weight
classes, such that the edges in $W_i$ appear before those in $W_{i-1}$,
for every $i$.  Using this sorted list of edges, we apply a greedy
algorithm for selecting a maximal matching, in which the current edge is added
to this matching if it remains a matching after doing so.  Then,
the post-processing time needed is linear in the
size of the memory used, that is, $ O( n \log_{\gamma} (\frac
n{\eps}))$. This concludes the presentation of the algorithm and its
implementation as a semi-streaming algorithm.

\paragraph{Analysis.} For purposes of analysis, we round down the weight of
each edge $e$ such that $w(e)\in W_i$ to be $\phi \gamma^i$. This
way, we obtain {\em rounded} edge weights. Now fix an optimal solution $\opt$ and denote by $\opt$ its weight, and by $\opt'$ its rounded weight. The next claim immediately follows from the definition of $W_i$.

\begin{lemma}
$\opt \leq \gamma \opt'$.
\end{lemma}

As an intermediate step, we analyze an improved algorithm
which keeps all weight classes.  That is, for each $i$, we use
$M_i$ to denote the maximal matching of class $W_i$ at the end of the
input, and denote by $M$ the solution obtained by this algorithm,
if we would have applied it. Similarly, we denote by $\opt_i$ the
set of edges in $\opt$ which belong to $W_i$.  For every $i$,
we define the set of vertices $P_i$, associated with $W_i$, to be
the set of endpoints of edges in $M_i$ that are not associated
with higher weight classes:
\[ P_i = \{ \,u,v \ |\ (u,v) \in M_i\} \setminus (P_{i+1} \cup P_{i+2}
\cup \cdots ). \]
For a vertex $p \in P_i$, we define its associated weight to be $\phi
\gamma^i$.  For vertices which do not belong to any $P_i$, we
let their associated weight be zero. We next bound the total
associated weight of all the vertices.

\begin{lemma}
The total associated weight of all the vertices is at most ${2\gamma\over \gamma-1} \cdot w(M) $.
\end{lemma}
\begin{proof}
  Consider a vertex $u \in P_i$ and let $(u,v)$ be the edge in $M_i$
  adjacent to $u$. If $(u,v) \in M$ then we charge the weight
  associated with $u$ to the edge $(u,v)$. Thus, an edge $e \in M_i$
  is charged at most twice from vertices associated with its own
  weight class. Otherwise, if $(u,v) \notin M$ then there must be
  some other edge $e \in M \cap M_j$, for some $j > i$, that prevented us
  from adding $(u,v)$ to $M$, in which case we charge the weight
  associated with $u$ to $e$. Notice that $u \notin e$, for otherwise,
  $u$ would not be associated with $W_i$. Thus, the edge $e \in M_j$
  must be of the form $e =(v,x)$ and can only be charged twice from
  vertices in weight class $i$, once through $v$ and once through
  $x$.


  To bound the ratio between $w(M)$ and the total associated weight of
  the vertices, it suffices to bound the ratio between the weight of
  an edge $e\in M$ and the total associated weight of the vertices
  which are charged to $e$.  Assume that $e\in M_j$, then there are at
  most two vertices which are charged to $e$ and class $i$ for all
  $i\leq j$, and no vertex is associated to $e$ and class $i$ for
  $i>j$.  Hence, the total associated weight of these vertices is at
  most
  \[ 2\sum_{i\leq j} \phi \gamma^i < 2\phi \gamma^j\cdot
  \sum_{i'=0}^{\infty} {1\over \gamma^{i'}}= 2\phi \gamma^j\cdot
  {1\over {1-1/\gamma}}=\phi \gamma^j \cdot {2\gamma\over
    \gamma-1},\] and the claim follows since $w(e)\geq \phi \gamma^j$.
\end{proof}

It remains to bound $\opt'$ with respect to the total associated weight.
\begin{lemma}
$\opt'$ is at most the total weight associated with all vertices.
\end{lemma}
\begin{proof}
  It suffices to show that for every edge $e=(x,y)\in \opt_i$ the
  maximum of the associated weights of $x$ and $y$ is at least the
  rounded weight of $e$.  Suppose that this claim does
  not hold, then $x$ and $y$ are not covered by $M_i$, as otherwise
  their associated weight would be at least $\phi \gamma^i$. Hence, when the algorithm considered $e$, we would have added $e$ to
  $M_i$, contradicting our assumption that $x$ and $y$ are not covered
  by $M_i$.
\end{proof}

Using the above sequence of lemmas, and recalling that we lose another
$\eps$ in the approximation ratio due to disregarding edges of weight at most $2\eps
w_{\max}\over n$, we obtain the following inequality:
\begin{equation} \label{eq1}
\opt \leq \gamma \opt' \leq \left(\gamma \cdot {2\gamma\over \gamma-1} +\eps\right) \cdot w(M).
\end{equation}
Therefore, we establish the following theorem.
\begin{theorem} \label{simple_det_theorem}
Our simple deterministic algorithm has an approximation ratio of $( {2\gamma^2\over
\gamma-1}+\eps)$. This ratio can be optimized to $8+\eps$ by picking $\gamma=2$.
\end{theorem}

The next example demonstrates that the analysis leading to Theorem \ref{simple_det_theorem} is tight.

\begin{example}
  Let $k$ be some large enough integer and $\eps > 0$ be
  sufficiently small. Consider the instance depicted in
  Figure~\ref{fig:tight-example-8}, where $M=M_k$ consists of a single
  edge $(x,y)$ with weight $\gamma^k$. For every $0 \leq i < k$, the
  matching $M_{i}$ consists of exactly two edges $(\alpha_i,x)$ and
  $(y,\beta_i)$ each of weight $\gamma^i$, and $\opt_i$ consists of
  two edges $(\alpha_i,a_i)$ and $(\beta_i,b_i)$ each of weight
  $\gamma^{i+1}-\eps$. In addition, there are two edges $(a_k, x)$ and
  $(b_k,y)$ whose weight is $\gamma^{k+1} - \eps$. It is easy to
  see that each $M_i$ is indeed maximal in its own weight class. Given
  these matchings, our greedy selection rule will output a single edge
  $(x,y)$ with total weight $\gamma^k$ (notice that computing a
  maximum weight matching in $M_0 \cup \cdots \cup M_k$ does not help
  when $\gamma \geq 2$).  Moreover, the value of the optimal solution
  matches our upper bound up to an additive $O(\eps)$ term.
\end{example}

\begin{figure}[htbp]
  \centering

  \beginpgfgraphicnamed{tight-example-8}
  \begin{tikzpicture}[xscale=2.5,yscale=2]
    \tikzstyle{every node}=[label distance=2pt,fill,inner sep=1.25pt]
    \tikzstyle{weight}=[fill=white,inner sep=3pt]

    \draw (2,2) node[label=below:$x$] (x) { } -- node[weight] {$\gamma^k$} +(1,0)
    node[label=below:$y$] (y) {} ;

    \foreach \y/\n/\m in {1/0/1,3/k-1/{k}}
    {
      \draw (0,\y) node[label=above:$a_{\n}$] {}  -- node[weight] {$\gamma^{\m} \! - \! \eps$}
            (1,\y) node[label=above:$\alpha_{\n}$] {} -- node[weight] {$\gamma^{\n}$}
            (x);
      \draw (y) -- node[weight]  {$\gamma^{\n}$}
            (4,\y) node[label=above:$\beta_{\n}$] {}  -- node[weight]
            {$\gamma^{\m} \! - \! \eps$}
            (5,\y) node[label=above:$b_{\n}$] {};
    };
    \draw[loosely dotted, line width=1pt] (0.5,1.75) -- (0.5,2.25);
    \draw[loosely dotted, line width=1pt] (4.5,1.75) -- (4.5,2.25);
    \draw (x) -- node[weight] {$\gamma^{k+1}\! -\!\eps$} +(0.2,1.3)
    node[label=above:$a_k$] {};
    \draw (y) -- node[weight] {$\gamma^{k+1}\! -\!\eps$} +(-0.2,1.3) node[label=above:$b_k$] {};
  \end{tikzpicture}
  \endpgfgraphicnamed

  \caption{ \label{fig:tight-example-8} A tight example for our
    deterministic algorithm. }

\end{figure}
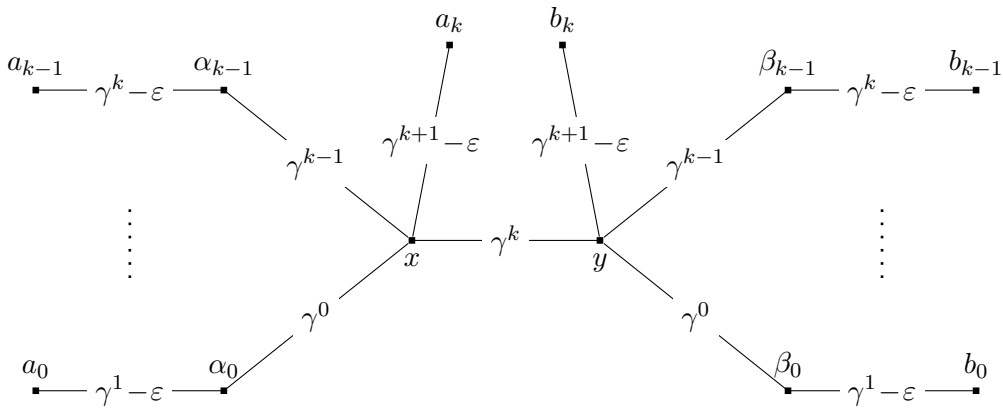

\subsection{Improved approximation ratio through randomization} \label{ran:alg}

In what follows, we analyze a randomized variant of the deterministic algorithm which was presented in the previous subsection. In general, this variant sets the value of $\phi$ to be $\phi=\gamma^{\delta}$ where $\delta$ is a random variable.  This method is commonly referred to as {\it randomized geometric grouping}.

Formally, let $\delta$ be a continuous random variable which is uniformly distributed on the interval $[0,1)$.  We define the weight class $W_i(\delta)=[\gamma^{i+\delta},\gamma^{i+1+\delta})$, and run the algorithm as in the previous subsection.  Note that this algorithm uses only the partition of the edges into classes and not the precise values of their weights.  In addition, we denote by $M(\delta)$ the resulting matching obtained by the algorithm, and by $TW(\delta)$ the total associated weight of the vertices, where for a vertex $p \in P_i$ we define its associated weight to be $\gamma^{i+\delta}$ (i.e., the minimal value in the interval $W_i(\delta)$). We also denote by $\opt'(\delta)$ the value of $\opt'$ for this particular $\delta$.

For any fixed value of $\delta$, inequality (\ref{eq1}) immediately implies $\opt'(\delta) \leq ( {2\gamma\over \gamma-1} +\eps )\cdot w(M(\delta)).$
Note that $\opt'(\delta)$ and $w(M(\delta))$ are random variables, such that for each realization of $\delta$ the above inequality holds.  Hence, this inequality holds also for their expected values.  That is, we have established the following lemma where $\expar{\delta}{\cdot}$ represents expectation with respect to the random variable $\delta$.
\begin{lemma}\label{lem2prime}
$\expar{\delta}{\opt'(\delta)} \leq ( {2\gamma\over \gamma-1} +\eps )\cdot \expar{\delta}{w(M(\delta))}$.
\end{lemma}

We next lower bound $\opt$ in terms of  $\expar{\delta}{\opt'(\delta)}$.
\begin{lemma} \label{lem:3prime}
$\frac{\gamma \ln \gamma}{\gamma-1}\cdot \expar{\delta}{\opt'(\delta)} \geq \opt$.
\end{lemma}
\begin{proof}
We will show the corresponding inequality for each edge $e\in \opt$.  We denote by $w'_{\delta}(e)$ the rounded weight of $e$ for a specific value of $\delta$.  Then, it suffices to show that $\frac{\gamma \ln \gamma}{\gamma-1}\cdot \expar{\delta}{w'_{\delta}(e)} \geq w(e)$.
 Let $p$ be an integer, and let $0\leq \alpha < 1$ be the value that satisfies $w(e)=\gamma^{p+\alpha}$. Then, for $\delta\leq \alpha$,
$w'_{\delta}(e)=\gamma^{p+\delta}$, and for $\delta
> \alpha$, $w'_{\delta}(e)=\gamma^{p-1+\delta}$, thus the expected rounded weight of $e$ over the choices of $\delta$ is
\[ \expar{\delta}{w'_{\delta}(e)} = \int_0^\alpha
\gamma^{p+\delta}d\delta+\int_\alpha^1
\gamma^{p-1+\delta}d\delta={1\over \ln \gamma}\cdot \left(
\gamma^p(\gamma^\alpha-1)+\gamma^{p-1}(\gamma-\gamma^\alpha)\right)=w(e)\cdot \left(1-\frac
1{\gamma}\right){1\over \ln \gamma}, \]
and the claim follows.
\end{proof}

Combining the above two lemmas we obtain that the expected  weight of the resulting solution is at least $(\frac {(\gamma-1)^2}{2\gamma^2\ln \gamma} + \eps) \cdot \opt$.  This approximation ratio is optimized for $\gamma \approx 3.513$, where it is roughly $(4.9108 + \eps)$. Hence, we have established the following theorem.

\begin{theorem}
The randomized algorithm has an approximation ratio of roughly $4.9108+\eps$.
\end{theorem}

\subsection{Derandomization} \label{final_alg}

Prior to presenting our de-randomization, we slightly modify the
randomized algorithm of the previous subsection.  In this variation,
instead of picking $\delta$ uniformly at random from the interval
$[0,1)$ we pick $\delta'$ uniformly at random from the discrete set
$\{ 0,{1\over q},{2\over q},\ldots ,{q-1\over q}\}$.  We apply the
same method as in the previous section where we replace $\delta$ by
$\delta'$.  Then, using Lemma \ref{lem2prime}, we obtain $\expar{\delta'}
{\opt'(\delta')} \leq( {2\gamma\over \gamma-1} +\eps
)\cdot \expar{\delta'}{w(M(\delta'))}$.  To extend Lemma
\ref{lem:3prime} to this new setting, we note that $\delta'$ can be obtained by first picking $\delta$ and then rounding it down to the largest number in $\{ 0,{1\over q},{2\over q},\ldots
,{q-1\over q}\}$ which is at most $\delta$.  In this way, we couple the
distributions of $\delta$ and $\delta'$.  Now consider the rounded
weight of an edge $e$ in $\opt$ in the two distinct values of $\delta$
and $\delta'$.  The ratio between the two rounded weight is at most
$\gamma^{1/ q}$.  Therefore, we establish that $\frac{\gamma \ln
  \gamma}{\gamma-1}\cdot \gamma^{1/q} \cdot \expar{\delta}{\opt'(\delta)} \geq \opt$.  Therefore, the resulting approximation ratio of the
new variation is $\frac
{2\gamma^{2+1/q}\ln \gamma}{(\gamma-1)^2}+\eps$.  By settinf $q$ to be
large enough (picking $q=\lceil \frac
1{\log_{\gamma}(\eps/5)}\rceil$ is enough), the resulting
approximation ratio is bounded by $\frac {2\gamma^{2}\ln
  \gamma}{(\gamma-1)^2}+2\eps$.

De-randomizing the new variation in the semi-streaming model is straightforward.  We simply run in parallel all $q$ possible outcomes of the
algorithm, one for each possible value of $\delta'$, and pick the best solution among the $q$ solutions we obtained. Since $q$
is a constant (for fixed values of $\eps$), the resulting
algorithm is still a semi-streaming algorithm whose performance
guarantee is $4.9108+2\eps$.  By scaling $\eps$ prior to applying the
algorithm, we establish the following result.

\begin{theorem}
  For any fixed $\eps > 0$, there is a deterministic one-pass
  semi-streaming $(4.9108+\eps)$-approximation algorithm for the
  weighted matching problem.  This algorithm processes each input edge in constant time and required $O(n)$ time at the end of the input to compute the final output.
\end{theorem}



\section{Online Preemptive Matching} \label{sec:lower_bound}

In this section, we established the following theorem.

\begin{theorem}
The competitive ratio of any deterministic preemptive
online algorithm is at least $\R\approx 4.967$, where $\R$ is the
unique real solution of the equation $x^3=4(x^2+x+1)$.
\end{theorem}

Recall that the algorithms of \cite{FeigenbaumKMSZ05} and
\cite{McGregor05} can be viewed as online preemptive algorithms; their competitive ratios are $6$ and $5.828$, respectively.

\paragraph{Definitions of some constants.} Let $\C=\R-\eps$ for some $\eps>0$ and assume that a deterministic online algorithm achieves a competitive ratio of at most
$\C'=\C-\eps$. We construct an input graph iteratively, and show
that after a finite number of steps, the competitive ratio is
violated.

In the construction of the input, all edge weights come from
two weight sequences. The main weight sequence is $w_1, w_2,
\ldots$, and an additional weight function is
$w'_2,w'_3,\ldots$. These sequences are defined as follows:
\begin{itemize}
\item $w_1=1$, and $w_{k+1}=\frac 1{2\C+1}
((\C^2+1)w_{k}-\C\sum_{i=1}^{k-1} w_{i})$ for $k\geq1$.

\item $w'_{k+1}=\frac 1{\C}((\C+1)w_{k+1}-w_k)$.
\end{itemize}%
The first sequence is defined for $k+1$
only as long as $w_{k-1} \geq w_{k-2}$.  As soon as
$w_{k}<w_{k-1}$, the sequence stops with $w_{k+1}$, and the
length of the sequence $w_i$ is $n=k+1$. We later show that such
a value $k$ must exist. Let $S_i=\sum_{j=1}^i w_j$ (and $S_0=0$).

\paragraph{Properties of the sequences.}
By definition, since $w'_{i+1}=w_{i+1}+\frac1{\C}(w_{i+1}-w_i)$,
if $w_{k}<w_{k-1}$, then $w_{k+1}<w_k$ holds as well. Note that
$w_i\leq w'_i$ for all $i<n-1$, by definition,
since $w_i\geq w_{i-1}$, but $w'_{n-1}<w_{n-1}$.
In addition, we have the following:
\[w'_{i+1}+w_{i+1}+S_{i-1}=\C w_i .\]
This equality holds for $i=1,2,\ldots ,n-2$ since
\begin{eqnarray*}
w'_{i+1}+w_{i+1}+S_{i-1} & = & {\C +1\over \C} w_{i+1}-{w_i\over \C} +
w_{i+1}+S_{i-1}\\
& = & {2\C+1 \over \C} \cdot {1\over 2\C+1} \cdot
((\C^2+1)w_i-\C S_{i-1}) +S_{i-1} - {w_i \over \C} \\
& = & \C w_i,
\end{eqnarray*}
where the first equality holds by definition of $w'_{i+1}$, the
second equality holds by definition of $w_{i+1}$, and the
third one by simple algebra. In addition,
\[S_{i-2}+w_{i}+w_{i+1}+w'_{i+1}=\C w'_i .\]
The last equality holds for $i=2,3,\ldots ,n-2$ since
\begin{eqnarray*}
S_{i-2}+w_{i}+w_{i+1}+w'_{i+1} & = & S_{i-2}+w_i+ {2\C +1\over \C}
w_{i+1}-{w_i\over \C} \\
& = & S_{i-2} + {\C-1\over \C}w_i + {2\C+1\over \C}\cdot \frac 1{2\C+1} ((\C^2+1)w_{i}-\C S_{i-1})) \\
& = & (\C+1)w_i +S_{i-2}-S_{i-1} \\
& = & (\C+1)w_i-w_{i-1} \\
& = & \C w'_i,
\end{eqnarray*}
where the first equality holds by definition of $w'_{i+1}$, the second by
definition of $w_{i+1}$, the third by simple algebra, the fourth
by definition of $S_{i-1}$ and $S_{i-2}$, and the last one by
definition of $w'_i$.

\paragraph{Input construction, step 1.}
To better understand our construction, we advice the reader to consult Figure~\ref{fig:preemptive-lowerbound}. The input is created in $n$ steps. In the initial step, two edges $(a_1,x_1)$ and $(b_1,x_1)$, each of weight
$w_1$, are introduced. Assume that after both edges have arrived,
the online algorithm holds the edge $(a_1,x_1)$. All future edges
either have endpoints which are new vertices, or in the set
$\{a_1,x_1\}$ (i.e., they do not contain $b_1$ as an endpoint).
An optimal solution keeps $(b_1,x_1)$.

\begin{figure}[htbp]
  \centering

  \beginpgfgraphicnamed{preemptive-lowerbound}
  \begin{tikzpicture}[scale=1.1,decoration={zigzag,amplitude=2pt,segment
      length=4pt}]
    \footnotesize
    \tikzstyle{round}=[inner sep=5pt,draw,fill=none,rectangle]
    \tikzstyle{vertex}=[fill,draw,circle,inner sep=1.25pt]

    \def\startpoint{(0,6)}
    \draw \startpoint node[round] {$i=1$};
    \draw \startpoint
          ++(0,-2) node[vertex,label=below:$b_1$] (b1) {} --
          ++(0,1) node[vertex,label=above:$x_1$] (x1) {};
    \draw \startpoint ++(1,-1) node[vertex,label=above:$a_1$] (a1) {};
    \draw [decorate]  (x1) -- (a1);

    \def\startpoint{(4.5,6)}
    \draw \startpoint node[round] {$i=2$};
    \draw \startpoint
          ++(0,-2) node[vertex,label=below:$b_1$] (b1) {} --
          ++(0,1) node[vertex,label=above:$x_1$] (x1) {} --
          ++(1,0) node[vertex,label=above:\parbox{1em}{$a_1$ \\ $x_2$}] (a1) {};
    \draw (a1) -- +(0,-1) node[vertex,label=below:$b_2$] {};
    \draw (a1) +(1,0) node[vertex,label=above:$a_2$] (a2) {};
    \draw [decorate] (a1) -- (a2);

    \def\startpoint{(10,6)}
    \draw \startpoint node[round] {$i=3$};
    \draw \startpoint
          ++(0,-2) node[vertex,label=below:$b_1$] (b1) {} --
          ++(0,1) node[vertex,label=above:$x_1$] (x1) {} --
          ++(1,0) node[vertex,label=above:\parbox{1em}{$a_1$ \\ $x_2$
            \\ $y_3$}] (a1) {};
    \draw (a1) -- +(0,-1) node[vertex,label=below:$b_2$] {};
    \draw (a1) -- +(1,0) node[vertex,label=above:\parbox{1em}{$a_2$\\$x_3$}] (a2) {};
    \draw (a2) -- +(0,-1) node[vertex,label=below:$b_3$] {};
    \draw (a2) -- +(1,0) node[vertex,label=above:$a_3$] (a3) {};
    \draw (a2) +(2,0) node[vertex,label=above:$c_3$] (c3) {};
    \draw (a1) edge [decorate,bend right] (c3);

    \def\startpoint{(0,2.5)}
    \draw \startpoint node[round] {$i=4$};
    \draw \startpoint
          ++(0,-2) node[vertex,label=below:$b_1$] (b1) {} --
          ++(0,1) node[vertex,label=above:$x_1$] (x1) {} --
          ++(1,0) node[vertex,label=above:\parbox{1em}{$a_1$ \\ $x_2$
            \\ $y_3$ \\$y_4$}] (a1) {};
    \draw (a1) -- +(0,-1) node[vertex,label=below:$b_2$] {};
    \draw (a1) -- +(1,0) node[vertex,label=above:\parbox{1em}{$a_2$\\$x_3$}] (a2) {};
    \draw (a2) -- +(0,-1) node[vertex,label=below:$b_3$] {};
    \draw (a2) -- +(1,0) node[vertex,label=above:$a_3$] (a3) {};
    \draw (a2) +(2,0) node[vertex,label=above:\parbox{1em}{$c_3$ \\ $x_4$}] (c3) {};
    \draw (a1) edge [bend right] (c3);
    \draw (c3) -- +(0,-1) node[vertex,label=below:$b_4$] {};
    \draw (c3) -- +(1,0) node[vertex,label=above:$a_4$] (a4) {};
    \draw (a2) +(4,0) node[vertex,label=above:$c_4$] (c4) {};
    \draw (a1) edge [decorate,bend right] (c4);

    \def\startpoint{(7,2.5)}
    \draw \startpoint node[round] {$i=5$};
    \draw \startpoint
          ++(0,-2) node[vertex,label=below:$b_1$] (b1) {} --
          ++(0,1) node[vertex,label=above:$x_1$] (x1) {} --
          ++(1,0) node[vertex,label=above:\parbox{1em}{$a_1$ \\ $x_2$
            \\ $y_3$ \\$y_4$}] (a1) {};
    \draw (a1) -- +(0,-1) node[vertex,label=below:$b_2$] {};
    \draw (a1) -- +(1,0) node[vertex,label=above:\parbox{1em}{$a_2$\\$x_2$}] (a2) {};
    \draw (a2) -- +(0,-1) node[vertex,label=below:$b_3$] {};
    \draw (a2) -- +(1,0) node[vertex,label=above:$a_3$] (a3) {};
    \draw (a2) +(2,0) node[vertex,label=above:\parbox{1em}{$c_3$ \\ $x_4$}] (c3) {};
    \draw (a1) edge [bend right] (c3);
    \draw (c3) -- +(0,-1) node[vertex,label=below:$b_4$] {};
    \draw (c3) -- +(1,0) node[vertex,label=above:$a_4$] (a4) {};
    \draw (a2) +(4,0) node[vertex,label=above:\parbox{1em}{$c_4$\\$x_5$}] (c4) {};
    \draw (a1) edge [bend right] (c4);
    \draw (c4) -- +(0,-1) node[vertex,label=below:$b_5$] {};
    \draw (c4) +(1,0) node[vertex,label=above:$a_5$] (a5) {};
    \draw (a5) edge [decorate] (c4);

  \end{tikzpicture}
  \endpgfgraphicnamed
  \vspace{2em}

  \caption{ \label{fig:preemptive-lowerbound} An example of five steps of the lower
    bound construction. The curved edges denote the edge kept by the online algorithm at each time.
    In the first two steps, the edges $(x_i,a_i)$ are chosen by the algorithm.
    In the third step $(x_3,a_3)$ is not chosen by
    the algorithm, so $(y_3, c_3)$ arrives next. In the fourth
    $(x_4,a_4)$ is not chosen by the algorithm, so $(y_4,c_4)$ arrives
    next. In the fifth step $(x_5, a_5)$ is chosen by the algorithm,
    so no further edges arrive in this step. }

\end{figure}
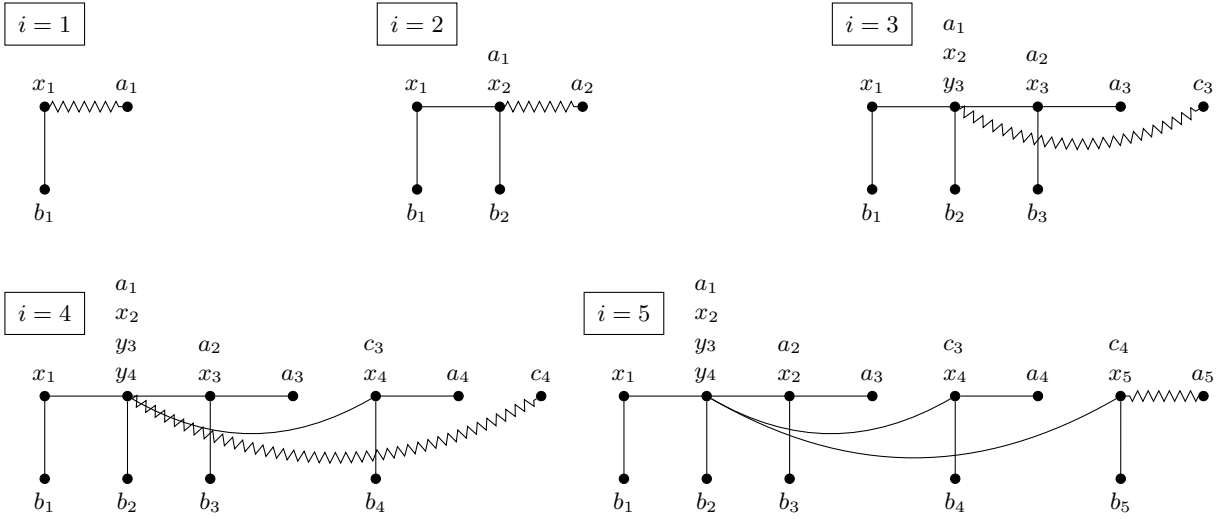

\paragraph{Input construction, properties.}
Every future step can be of two distinct types, which will be described
later on. Among the edges introduced below, vertices called $b_j$
denote endpoints which occur each on a single edge.

After step $i$, the following invariants are maintained. The
algorithm keeps a single edge denoted by $e_i$. If $i=1$, then
$e_i=(a_i,x_i)$. If $i>1$, then this edge can be one of two
edges, $(a_i,x_i)$ or $(c_i,y_i)$. If $e_i=(a_i,x_i)$, then its
weight is $w_i$, and an optimal solution has one edge of each
weight $w_1,w_2,\ldots,w_{i-1},w_i$. No future edges will have
common endpoints with these $i$ edges, except, possibly, with the
endpoint $x_i$ of the edge of weight $w_i$ (the edge of this
weight which this optimal solution keeps is always $(x_i,b_i)$).
 Otherwise, $e_i=(c_i,y_i)$,
and its weight is $w'_i$, in which case an optimal solution can
have edges of weights $w_1,w_2,\ldots,w_i$, except for one weight
$w_j$ for some $j<i$. This index $j$ is used in the definition of
the next step, and the properties of the current step. In
addition to these $i-1$ edges, the optimal solution also has the
edge $(c_i,y_i)$. Future edges will have endpoints which are new
vertices, or in the set $\{c_i,y_i\}$. In the last case, the
vertex $y_i$ is equal to the vertex $x_j$. The invariants clearly
hold after the first step. We next define all other steps and
show that the invariants hold for each option.


\paragraph{Input construction, step $\bs{n}$.}
If $i+1=n$, the last step consists of an edge of weight $w_n$. Let
$x_{n}=a_{n-1}$, if $e_{n-1}=(a_{n-1},x_{n-1})$ and otherwise
$x_{n}=c_{n-1}$. The new edge is $(x_n,b_n)$, where $b_n$ is a
new vertex. This edge has a common endpoint with the edge that
the algorithm has. In fact, the algorithm has an edge of weight
at least $w_{n-1}>w_n$, and thus we assume that it does not
preempt it. If the algorithm has an edge of weight $w_{n-1}$, the
edge $(x_n,b_n)$ does not have $x_i$ as an endpoint, so adding
the new edge to the optimal solution does not require the removal
of any edges, and the profit of the optimal solution is $S_n$. If
the algorithm has an edge of weight $w'_{n-1}$, the new edge is
$(c_{n-1},b_n)$. We replace the edge $(c_{n-1},y_{n-1})$  of the
optimal solution by the new edge. In addition, the edge
$(y_{n-1},b_j)=(x_j,b_j)$ (where $j$ is the index such that the
optimal solution before the modification of the current step does
not have an edge of weight $w_j$) is added to the optimal
solution, since the endpoint $y_{n-1}$ became free, and the
endpoint $b_j$ only has degree 1. The profit of the optimal
solution is $S_n$ again. Recall that $w'_{n-1}\leq w_{n-1}$, and
hence the algorithm earns (in both cases) at most $w_{n-1}$. Note
also that the optimal solution has value of $S_n$ and if $w_n<0$
then we can drop the edge of this weight from the optimal
solution and get a solution of value $S_{n-1}$. Therefore, we
will use $S_{n-1}$ as a lower bound on the value of the optimal
solution in this case. Thus we will show later that
$\frac{S_{n-1}}{w_{n-1}} \geq \C > \C' $.

\paragraph{Input construction, step $\bs{i+1}$, for $\bs{i+1<n}$.}
We next show how to construct the edges of step $i+1$, for the
case $i+1<n$. We introduce two new edges of weight $w_{i+1}$. Let
$x_{i+1}=a_i$, if $e_i=(a_i,x_i)$ and otherwise $x_{i+1}=c_i$.
The new edges are $(x_{i+1},b_{i+1})$, and $(x_{i+1},a_{i+1})$,
where $a_{i+1}$ and $b_{i+1}$ are new vertices. Both these edges
have a common endpoint with the edge that the algorithm has, and
the algorithm can either preempt the edge it has, in which case
we assume (without loss of generality) that it now has
$(x_{i+1},a_{i+1})$, or else it keeps the previous edge. If the
algorithm keeps the previous edge, let $y_{i+1}=x_i$, if
$e_i=(a_i,x_i)$ and otherwise $y_{i+1}=y_i$. In this case a third
edge, $(y_{i+1},c_{i+1})$, which has a weight of $w'_{i+1}$, is
introduced. The vertex $c_{i+1}$ is new.

There are four cases to consider. In the first case, if the
algorithm replaces the edge $(a_i,x_i)$ with the edge
$(x_{i+1},a_{i+1})=(a_i,a_{i+1})$, then an optimal solution can
add the edge $(x_{i+1},b_{i+1})$ to its edges, since the endpoint
$b_{i+1}$ is new, and the endpoint $a_i$ was introduced in the
previous step, in which the optimal solution obtained the edge
$(x_i,b_i)$.

If the algorithm replaces the edge $(c_i,y_i)$ with the edge
$(x_{i+1},a_{i+1})=(c_i,a_{i+1})$, an optimal solution can remove
the edge $(c_i,y_i)$ from its solution and add the two edges
$(x_{i+1},b_{i+1})=(c_i,b_{i+1})$ and $(y_i,b_j)=(x_j,b_j)$
(where $j$ is the index such that the optimal solution before the
modification of the current step does not have an edge of weight
$w_j$). This is possible since the endpoints $b_{i+1}$ and $b_j$
do not have other edges, and the endpoints $c_i$ and $y_i$ become
free.

In the last two cases, the invariants hold.  For the remaining two
cases note that if $w'_i\leq 0$ or $w_i<0$ and the algorithm has a
single edge of weight $w'_i$ or $w_i$, respectively, then the optimal
solution is strictly positive and the value of the algorithm is
non-positive, and hence the resulting approximation ratio in this case
is unbounded.  Hence, we can assume without loss of generality that if
the algorithm has a single edge at the end of step $i$, then its
weight is strictly positive.

If the algorithm does not replace the edge $(a_i,x_i)$ with the
edge $(x_{i+1},a_{i+1})=(a_i,a_{i+1})$, we show that it must
replace it with the edge $(y_{i+1},c_{i+1})=(x_i,c_{i+1})$. Assume
that this is not the case. Then the profit of the algorithm is
$w_i$ and the optimal solution can omit its edge $(x_i,b_i)$ and
add the edges $(x_i,c_{i+1})$ and $(a_i,b_{i+1})$ (since all these
endpoints are introduced in steps $i$ and $i+1$, except
for $x_i$, which becomes free). Thus the profit of the optimal
algorithm is $S_{i-1}+w_{i+1}+w'_{i+1}=\C\cdot w_i$, while the
profit of the online algorithm is $w_{i}$. Thus, the algorithm
must switch to the edge $(x_{i+1},a_{i+1})$, and the structure of
the optimal solution is according to the invariants.

If the algorithm does not replace the edge $(c_i,y_i)$ with the
edge $(x_{i+1},a_{i+1})=(c_i,a_{i+1})$, we show that it must
replace it with the edge $(y_{i+1},c_{i+1})=(y_i,c_{i+1})$. Assume
that this is not the case. Then the profit of the algorithm is
$w'_i$ and the optimal solution can omit its edge $(c_i,y_i)$ and
add the edges $(c_i,a_{i+1})$ and $(y_i,c_{i+1})$ (since $c_i$ and
$y_i$ become free, and the other two  endpoints are introduced in step
$i+1$). Thus the profit of the optimal algorithm is
$S_{i}-w_j+w_{i+1}+w'_{i+1}$, where $j\leq i-1$ and $i \geq 2$,
since $w_j\leq w_{j+1}\leq \cdots \leq w_{i-1}$ as $i-1\leq n-2$, we get that the
optimal profit is at least $S_{i-2}+w_{i}+w_{i+1}+w'_{i+1}=\C
w'_i$, while the profit of the online algorithm is $w'_{i}$. Thus,
the algorithm must switch to the edge $(x_{i+1},a_{i+1})$, and the
structure of the optimal solution is according to the invariants.

\paragraph{Bounding the competitive ratio.} We next define a recursive formula for $S_i$. By the
definition of the sequence $w_i$, we have
\begin{equation}\label{rec}
\left\{
\begin{array}{l}
S_0=0 \\
S_1=1 \\
S_{k+1}=\frac{\C^2+2\C+2}{2\C+1}S_{k}-\frac{\C^2+\C+1}{2\C+1}S_{k-1}, \quad \mbox{for } k \geq 1
\end{array} \right.
\end{equation}
We first use this recurrence to show that if $w_{n-1}<w_{n-2}$ then ${S_{n-1}\over w_{n-1}} \geq \C$.  To see this note that by assumption $S_{n-1}-S_{n-2}< S_{n-2}-S_{n-3}$, hence using the recurrence formula we conclude that
\[ S_{n-1}-2S_{n-2}+{2\C+1 \over -\C^2-\C-1} S_{n-1} +{\C^2+2\C+2\over \C^2+\C+1}S_{n-2}<0, \]
that is,
\[ S_{n-1}\cdot (\C^2+\C+1-2\C-1) +S_{n-2}\cdot (\C^2+2\C+2-2\C^2-2\C-2)<0 , \]
which is equivalent to $(\C^2-\C)S_{n-1}-\C^2S_{n-2}<0$, so $\C(S_{n-1}-S_{n-2})<S_{n-1}$, and we conclude that $\C w_{n-1} < S_{n-1}$, as we argued. Therefore, it remains to show that there is a value of $n$ such that $w_{n-2}>w_{n-1}$. To establish this claim, it suffices to show that there is a value of $j$ for which $w_j<0$ (since $w_1>0$). To prove this last claim, we will show that there is a value of $k$ such that $S_k<0$. Finally, to show the existence of such $k$, we will solve the linear homogeneous recurrence formula, and use the explicit form of $S_k$ to show that there is a value of $k$ such that $S_k<0$.

To solve the recurrence formula (\ref{rec}), we guess solutions of the form $S_k=x^k$ for all $k$, and get the following quadratic equation for $x$:
\[ (2\C+1)x^2-(\C^2+2\C+2)x+(\C^2+\C+1)=0 . \]%
We solve this quadratic equation and get its solutions
\begin{eqnarray*}
x_{1,2}&=&{ (\C^2+2\C+2) \pm \sqrt{(\C^2+2\C+2)^2-4(2\C+1)(\C^2+\C+1)}\over 2(2\C+1)}\\
&=& { (\C^2+2\C+2) \pm \sqrt{\C^4+4\C^2+4+4\C^3+8\C+4\C^2-8\C^3-4\C^2-8\C^2-4\C-8\C-4}\over 2(2\C+1)}\\
&=& { (\C^2+2\C+2) \pm \sqrt{\C(\C^3-4\C^2-4\C-4)}\over 2(2\C+1)}.
\end{eqnarray*}
Note that using $\C < \R$, and recalling that $\R$ is the
unique real solution of the equation $x^3=4(x^2+x+1)$, we conclude that
$\C(\C^3-4\C^2-4\C-4)<0$ and hence the two solutions are complex
numbers whose imaginary parts are not zero.  Since we got two
distinct solutions of $x$, it is known that the recurrence
formula (\ref{rec}) is solved by a formula of the form
$S_i=\alpha x_1^i +\beta x_2^i$ where $\alpha$ and $\beta$ are
constants.  We find the value of $\alpha$ and $\beta$ using the
conditions $S_0=0$ and $S_1=1$.  So we get the following set of
two equations: $\alpha+\beta=0$ (corresponding to $S_0=0$), and
$\alpha x_1+\beta x_2 =1$ (corresponding to $S_1=1$).  From the
first equation we conclude that $\beta=-\alpha$, and using this
we obtain $\alpha={1\over x_1-x_2}={2\C+1 \over
\sqrt{\C(\C^3-4\C^2-4\C-4)}}$.  Hence, the closed form solution
of $S_j$ for values of $\C <\R$ is as follows.
\begin{eqnarray}
S_j&=& {2\C+1 \over \sqrt{\C(\C^3-4\C^2-4\C-4)}} \left( {
(\C^2+2\C+2) + \sqrt{\C(\C^3-4\C^2-4\C-4)}\over 2(2\C+1)}
\right)^j \nonumber\\ && - {2\C+1 \over
\sqrt{\C(\C^3-4\C^2-4\C-4)}} \left( { (\C^2+2\C+2) -
\sqrt{\C(\C^3-4\C^2-4\C-4)}\over 2(2\C+1)} \right)^j \ .
\label{closed}
\end{eqnarray}
We use the notation $i=\sqrt{-1}$, and let $\alpha = A\cdot i$. As
noted above $\C(\C^3-4\C^2-4\C-4)<0$, and hence $A$ is a real
number.  We also define $r$ and $\theta$ such that ${
(\C^2+2\C+2) + \sqrt{\C(\C^3-4\C^2-4\C-4)}\over 2(2\C+1)}=
r(\cos(\theta)+i\sin(\theta))$, and also ${ (\C^2+2\C+2) -
\sqrt{\C(\C^3-4\C^2-4\C-4)}\over 2(2\C+1)}=
r(\cos(\theta)-i\sin(\theta))$, then we get the following formula
for $S_j$.
\begin{eqnarray*}
S_j&=&A\cdot i \cdot \left( r^j (\cos(\theta)+i\sin(\theta))^j-r^j(\cos(\theta)-i\sin(\theta))^j \right)\\
&=&A\cdot i \cdot \left( r^j (\cos(j\theta)+i\sin(j\theta))-r^j(\cos(j\theta)-i\sin(j\theta)) \right)\\
&=& A\cdot i \cdot r^j \cdot 2i\sin(j\theta) \\
&=&-2A r^j \sin(j\theta) \ .
\end{eqnarray*}
Note that $r^j >0$ for all $j$, and hence to show that the sequence $\{ S_j\}$ changes its sign as we required, it suffices to show that the sequence $\{ \sin(j\theta) \}$ changes its sign, but this last claim holds because $0<\theta <\pi$ (as the solutions $x_1$ and $x_2$ are not real numbers).  Hence, the claim follows.

\bibliographystyle{abbrv}


\end{document}